\documentclass[a4paper,USenglish,thm-restate]{lipics-v2021}

\AtBeginDocument{%
  }

\usepackage{amsmath}
\usepackage{amsthm}
\usepackage{amsfonts}
\usepackage{graphicx}
\usepackage{todonotes}
\usepackage{lipsum}
\usepackage{float}
\usepackage{thmtools}
\usepackage{placeins}

\usepackage{mdframed}
\newmdenv[linecolor=black,linewidth=0.5pt,roundcorner=5pt]{problem}

\newcommand{\graph}{Erd\H{o}s--R\'enyi graph $G(n, p)$}
\newcommand{\jgraph}{Erd\H{o}s--R\'enyi graph}
\newcommand{\tmix}{\tau_{mix}}

\date{}



\author{Anton Paramonov}{ETH Zurich}{aparamonov@ethz.ch}{https://orcid.org/0009-0000-0760-8746}{}
\author{Roger Wattenhofer}{ETH Zurich}{wattenhofer@ethz.ch}{https://orcid.org/0000-0002-6339-3134}{}

\authorrunning{Paramonov \& Wattenhofer}
\Copyright{Anton Paramonov and Roger Wattenhofer}

\acknowledgements{The authors thank Diana Ghinea and Bernhard Haeupler for their helpful feedback, which helped improve this paper.}

\SeriesVolume{123}
\ArticleNo{1}
\hideLIPIcs
\nolinenumbers

\begin{document}
    \title{Broadcast in Almost Mixing Time}
    \maketitle
    
    \begin{abstract}
        We study the problem of broadcasting multiple messages in the CONGEST model. In this problem, a dedicated source node $s$ possesses a set $M$ of messages with every message of size $O(\log n)$ where $n$ is the total number of nodes. The objective is to ensure that every node in the network learns all messages in $M$. The execution of an algorithm progresses in rounds, and we focus on optimizing the round complexity of broadcasting multiple messages. 
        
        Our primary contribution is a randomized algorithm for networks with expander topology. The algorithm succeeds with high probability and achieves a round complexity that is optimal up to a factor of the network's mixing time and polylogarithmic terms. It leverages a multi-COBRA primitive, which uses multiple branching random walks running in parallel. A crucial aspect of our method is the use of these branching random walks to construct an optimal (up to a polylogarithmic factor) tree packing of a random graph, which is then used for efficient broadcasting.
                    
        We also prove the problem to be NP-hard in a centralized setting and provide insights into why lower bounds that can be matched in expanders, namely graph diameter and $\frac{|M|}{\textit{minCut}}$, cannot be tight in general graphs.
    \end{abstract}

    \newpage
    \pagenumbering{arabic}  %
    \setcounter{page}{1}

    \section{Introduction}
    The CONGEST model \cite{peleg2000distributed} was originally introduced to study computer networks with limited communication channels. Though later, despite its simplicity, it has given rise to a wide range of theoretical challenges. Progress on those has not only allowed for faster distributed algorithms, but also led to a deeper understanding of fundamental graph abstractions (see, e.g., \cite{ghaffari2016distributedPlanar,ghaffari2016distributed,elkin2017distributed, rozhovn2020polylogarithmic,brandt2016lower}).

    The problem we address in this work fits precisely into this tradition. We study the broadcast of multiple messages, where a single source node holds a collection of messages that must be disseminated so that every node in the network receives them as quickly as possible. In real-world systems, this setting naturally models the distribution of data chunks in peer-to-peer file sharing or block chunks dissemination in blockchain protocols. Theoretically, however, the problem reduces to the tree packing in the underlying network graph - a concept still not fully understood.
    
    We show that this problem admits a fast solution in expander networks. In particular, we present an algorithm that broadcasts multiple messages with near-optimal round complexity: the overhead depends only on the network’s mixing time, up to a small polylogarithmic factor. The key technical contribution enabling this result is a nearly-optimal distributed tree packing procedure for a random graph, which serves as the structural backbone of our dissemination strategy and may be of independent interest beyond this application.
    
    We now proceed to define the setting formally.

    \subsection{Model and Problem}
    The CONGEST model is defined as follows. The network is modeled as a graph with $n$ nodes, where execution progresses in synchronous rounds. In each round, a node can send a message of size $O(\log n)$ bits to each of its neighbors. Importantly, nodes do not have prior knowledge of the network topology.

    Although the CONGEST model has been extensively studied over the past two decades, the fundamental problem of broadcasting multiple messages remains unsolved for general topologies.
    
    \begin{definition}[Multi-message broadcast]
    A dedicated source node $s$ possesses a set $M$ of $k$ messages, where each message $m \in M$ has a size of $O(\log n)$ bits. The objective is to ensure that every node in the network learns all messages in $M$.
    \end{definition}

    As pointed out by Ghaffari \cite{ghaffari2015distributed}, the problem suggests an $\Omega(D + k)$ round complexity lower bound, where $D$ is the diameter of the graph. For example, consider a path graph with $s$ as its first node. Any algorithm would require at least $D + k - 1$ rounds to transmit all messages to the last node. However, this bound is only \emph{existential}, meaning there exists a graph for which $\Omega(D + k)$ rounds are needed. In contrast, consider a complete graph with $k = n$. Here, broadcasting can be completed in two rounds: in the first round, $s$ sends the $i$-th message to the node $i$, and in the second round, each node broadcasts the message it received in round $1$. This is significantly better than the $\Omega(k) = \Omega(n)$ bound suggested by the path graph example. These contrasting cases highlight the importance of algorithms that adapt to the underlying topology. Our paper presents such an algorithm, achieving \emph{universal optimality} \cite{garay1998sublinear} on expander graphs. Specifically, it completes the multi-message broadcast on every expander $G$ in a number of rounds within a small overhead of the best possible for $G$. Before stating our results formally, we introduce some necessary terminology.

    \subsection{Preliminaries and Notation}
    Throughout the paper, for a graph $G$, $E(G)$ is the edge set, $V(G)$ is the vertex set, $D(G)$ is the diameter, $d_G(v)$ is the degree of a node $v$ in $G$, $\delta(G)$ is the smallest vertex degree, and $\Delta(G)$ is the largest vertex degree. We do not explicitly specify $G$ if it is obvious from the context, e.g., we can write $\delta$ instead of $\delta(G)$. With high probability (w.h.p.) means with a probability of at least $1 - O(\frac{1}{n^C})$ for some constant $C > 0$, with the probability being taken over both the randomness of the graph (when we assume random graphs) and the random bits of the algorithm. We assume that $\Tilde{O}$ and $\Tilde{\Omega}$ hide $polylog(k, n)$ factors.

    In multiple places in this paper, we are using classical Chernoff bounds. A reader can find formal statements in Appendix \ref{sec:appendix: techincal defintions}.

    One of the key graph-theoretical components in our approach is \emph{tree packing}. A tree packing of a graph $G$ is a collection of spanning subtrees of $G$. The tree packing is characterized by three parameters: (1) its size $S$, i.e., the number of trees, (2) its diameter $H$, i.e., the maximal diameter of a tree, and (3) its weight $W$, i.e., the maximal number of trees sharing a single edge.

    The present work focuses specifically on two graph families, namely {\jgraph}s and expanders. An \emph{\graph} is a graph on $n$ vertices where each edge exists independently from others with probability $p$ \cite{erd6s1960evolution}. Throughout the paper, let $C_p$ denote a sufficiently large constant\footnote{It suffices to take $C_p = 2700$. It was not a concern for the present work to optimize this constant.} such that $G(n, p)$ is connected w.h.p. for $p \geq \frac{C_p\log n}{n}$. The condition $p = \Omega(\frac{\log n}{n})$ is necessary, since for $p \leq \frac{\log n}{n}$, there is a constant probability that the graph is disconnected \cite{erd6s1960evolution}.
    
    We refer to a graph as an \emph{expander} if it has small (polylogarithmic) mixing time\footnote{An alternative way is to say that expander graphs are those that feature an inverse polylogarithmic \emph{conductance}; these two notions are equivalent. See Appendix \ref{sec:appendix: techincal defintions} for a definition of conductance.}. For our purposes, it will be convenient to define the mixing time of an undirected graph by reconsidering it as being bidirected, that is, with each undirected edge $(u, v)$ replaced with $(u, v)$ and $(v, u)$. The \emph{mixing time} \( \tmix \) of a bidirected graph is defined as follows. Consider a lazy random process that starts at an arbitrary edge of the graph. At each step, with probability \( \frac{1}{2} \), the process remains at the current edge, and with probability \( \frac{1}{2} \), it transitions to a uniformly random adjacent edge (a directed edge $e_2$ is adjacent to a directed edge $e_1$ if they are of the form $e_1 = (u, v)$ and $e_2 = (v, w)$). It is known that this process admits a stationary distribution \( \pi \), which is uniform over all edges. Moreover, regardless of the starting edge, the distribution \( D_t \) of the walk after \( t \) steps converges to \( \pi \). We define \( \tmix \) as the smallest \( t \) for which \( D_t \) is inversely polynomial close to \( \pi \). For a formal definition, please refer to Appendix \ref{sec:appendix: techincal defintions}. We are now ready to formally state our results.

    \subsection{Our Contribution}
    Our main result is an algorithm to solve the multi-message broadcast problem with only an overhead of $\tilde{O}(\tmix)$. We highlight that our algorithm can be run on \emph{any} graph, but it will only be efficient on graphs with small $\tmix$, i.e., expanders.
    \begin{theorem}
        There exists a randomized distributed algorithm that for any graph $G$ solves the multi-message broadcast problem in $O(\log^3n \cdot \tmix \cdot \mathrm{OPT})$ rounds with high probability, where $\tmix$ is the mixing time of $G$ and $\mathrm{OPT}$ is the optimal round complexity for the given problem instance.
    \end{theorem}

    In the analysis, we use $k/\delta(G)$ as a natural lower bound for $\mathrm{OPT}$. Indeed, if a node has degree $\delta(G)$, it requires at least $k/\delta(G)$ rounds to receive all $k$ messages. While this bound is meaningful on expander graphs, we show in Section~\ref{sec:bad lbs} that it can be far from tight even on graphs with constant diameter, which highlights the necessity of exploiting expansion properties in our approach.

    We design the algorithm for an arbitrary expander by building on the algorithm for random graphs, which achieves near-optimal performance in an important special case when the network is modeled as an {\graph}.

    \begin{restatable}{theorem}{thmRandomMain}
        \label{thm:random main}
        For an {\graph} with $p \geq \frac{C_p\log n}{n}$, there exists a distributed randomized algorithm that completes the broadcast in $O(\log^2n + \log n \cdot \frac{k}{\delta(G)})$ rounds w.h.p.
    \end{restatable}
    
    To obtain Theorem \ref{thm:random main}, we use the following result of independent interest, which contributes to a line of work \cite{ghaffari2013distributed, chuzhoy2020packing, ghaffari2015distributed, censor2014distributed, gao2014arboricity} on low-diameter tree packing:
    
    \begin{restatable}{theorem}{thmTreePacking}
        \label{thm:building tree packing}
        For an {\graph} with $p \geq \frac{C_p\log n}{n}$, there exists a distributed randomized algorithm that produces a tree packing of size $\delta(G)$, diameter $O(\log n)$, and weight $O(\log n)$ w.h.p.
    \end{restatable}
    
    We construct the latter algorithm by utilizing multiple Coalescing Branching Random Walks \cite{cooper2016coalescing} that run in parallel. To the best of our knowledge, this technique has never been used before in the context of distributed algorithms.
    
    Finally, to map the terrain of the problem, we prove the hardness result in the centralized setting. Namely, we show that computing the exact number of rounds required for multi-message broadcast is NP-hard on general graphs.

    \section{Related Work}
    \paragraph*{Previous Work.} The first work to address universal optimality for the multi-message broadcast problem in the CONGEST model was by Ghaffari \cite{ghaffari2015distributed}. In that paper, the algorithm consists of two phases: (1) constructing a \emph{tree packing}, and (2) performing the broadcast using the constructed tree packing. With a tree packing of diameter $H$, size $S$, and weight $W$, one can complete a multi-message broadcast in $O((H + \frac{k}{S}) \cdot W)$ rounds by splitting messages uniformly across the trees and propagating them sequentially within each tree. However, the limitation of \cite{ghaffari2015distributed} is that constructing the tree packing requires $\tilde{\Omega}(D + k)$ rounds, preventing the approach from achieving universal optimality. Observe that the problem can be solved in $O(D + k)$ rounds by a naive strategy: first construct a BFS tree rooted at $s$, and then downcast messages in it one by one. Thus, the result of \cite{ghaffari2015distributed} yields an improvement only when multiple consecutive instances of the problem are solved, allowing the precomputed tree packing to be reused.

    A subsequent work by Ghaffari et al. \cite{chandra2024fast} considered the tree packing approach on highly connected graphs, i.e., graphs with high edge connectivity $\lambda$. The primary result of this work is an algorithm that runs in $\Tilde{O}(\frac{n + k}{\lambda})$ rounds. This complexity is optimal when $k = \Omega(n)$, as $\frac{k}{\lambda}$ represents an information-theoretic lower bound. However, the algorithm may incur a $\Tilde{\Omega}(n)$ factor overhead in cases where $\lambda$ and $k$ are small compared to $n$.

    Notably, both \cite{ghaffari2015distributed} and \cite{chandra2024fast} consider a slightly more general problem where initially $M$ is not necessarily known to a single node, but different nodes can possess disjoint subsets of $M$. We adhere to our version, where $M$ is initially held by a single node, as it simplifies the presentation. Importantly, when a tree packing is available, the multiple-source version can be reduced to the single-source version without increasing the round complexity (see Remark~\ref{rem:single source}).
    
    \paragraph*{Routing.} One fundamental information-dissemination problem in distributed computing is \emph{routing}, where the goal is to deliver a set of messages from source nodes to their respective destination nodes. Unlike \emph{broadcast}, which involves sending the same message(s) to all nodes, routing requires sending individual message(s) for each source-destination pair.

    Ghaffari et al. \cite{ghaffari2017distributed} approach this problem for expanders by constructing a hierarchy of recursively embedded {\jgraph s}, and achieve routing in $\tau_{\text{mix}} \cdot 2^{O(\sqrt{\log n \log\log n})}$ rounds. Note that $2^{O(\sqrt{\log n \log\log n})}$ dominates $\log^c n$ for any constant $c$. This result was further improved in \cite{ghaffari2018new}, which reduces the round complexity to $O(2^{\sqrt{\log n}})$. Subsequently, Chang and Thatchaphol \cite{chang2020deterministic} presented a deterministic version of expander routing, matching (up to polylogarithmic factors) the round complexity of the randomized algorithm by Ghaffari et al.
    
    For general graphs, Haeupler et al. \cite{haeupler2022hop} provide a routing algorithm that runs in $\mathrm{poly}(D) \cdot n^{o(1)}$ rounds. Their approach leverages \emph{expander decomposition} and \emph{hop-constrained expanders}—subgraphs with small diameter and strong expansion properties. In fact, \cite{haeupler2022hop} obtain a stronger result: given that for every source-sink pair $(s_i, t_i)$, the source $s_i$ is at most $h$ hops from its destination $t_i$, routing can be completed in $O(D + \mathrm{poly}(h)) \cdot n^{o(1)}$ rounds.

    Unlike the above approaches, our paper provides an algorithm with only a polylogarithmic overhead.

    \paragraph*{Network Information Flow.} The network information flow problem \cite{ahlswede2000network} is defined as follows. The network is a directed graph $G(V, E)$ with edge capacities, a source node $s \in V$, and sink nodes $T \subseteq V$. The question is: at what maximal rate can the source send information so that all of the sinks receive that information at the same rate? In the case of a single sink $t$, the answer is given by the $\text{max-flow}(s, t)$. However, when there are multiple sinks $T$, the value $\min\limits_{t \in T} \text{max-flow}(s, t)$ may not be achievable if nodes are only allowed to relay information. In fact, the gap can be as large as a factor of $\Omega(\log n)$ \cite{jaggi2005polynomial}. Nevertheless, if intermediate nodes are allowed to send (linear \cite{li2003linear}) codes of the information they receive, then $\min\limits_{t \in T} \text{max-flow}(s, t)$ becomes achievable \cite{ahlswede2000network}. Notably, in the specific case where $T = V \setminus \{s\}$ (the setting considered in the present paper), the rate of $\min\limits_{t \in T} \text{max-flow}(s, t)$ becomes achievable without coding \cite{wu2004comparison}. The decentralized version of network information flow was studied in \cite{ho2003benefits, fragouli2004decentralized, ho2011universal}. The most relevant work in this direction is  \cite{swamy2013asymptotically} by Swamy et al., where the authors establish an optimal algorithm for the case of random graphs whose radius is almost surely bounded by $3$. Our approach works for general expanders, and in the case of random graphs, allows an expected radius to grow infinitely with $n$ (see \cite{chung2001diameter} for analysis of the diameter of a random graph).
    
    The key difference between the network information flow problem and the multi-message broadcast in CONGEST is that in our problem, the focus is on round complexity, whereas in the information flow problem, the solution is a "static" assignment of messages to edges, and the focus is on throughput.

    \paragraph*{Tree Packing.} The problem of tree packing has been extensively studied, as summarized in the survey by Palmer \cite{palmer2001spanning}. Foundational results in this area include those by Tutte \cite{tutte1961problem} and Nash-Williams \cite{nash1961edge}, who demonstrated that an undirected graph with edge connectivity $\lambda$ contains a tree packing of size $\lfloor \frac{\lambda}{2} \rfloor$. Edmond \cite{edmonds1972} extended this result to directed graphs, showing that such graphs always contain $\lambda$ pairwise edge-disjoint spanning trees rooted at a sender $s \in V$, where $\lambda$ is the minimum number of edges that must be removed to make some node unreachable from $s$. However, these results do not address the diameter of the tree packing.

    Chuzhoy et al. \cite{chuzhoy2020packing} tackled the challenge of finding tree packings with small diameter. They presented a randomized algorithm that, given an undirected graph with edge connectivity $\lambda$ and diameter $D$, outputs with high probability a tree packing of size $\lfloor \frac{\lambda}{2} \rfloor$, weight $2$, and diameter $O((101k \log n)^D)$.

    Tree packing on random graphs was studied by Gao et al. \cite{gao2014arboricity}, who showed that asymptotically almost surely, the size of a spanning tree packing of weight $1$ for a {\graph} is $\min\left\{\delta(G), \frac{|E(G)|}{n - 1}\right\}$, which corresponds to two straightforward upper bounds.

    In the CONGEST model, tree packing was investigated by Censor-Hillel et al. \cite{censor2014distributed}. They proposed an algorithm to decompose an undirected graph with edge connectivity $\lambda$ into fractionally edge-disjoint weighted spanning trees with total weight $\lceil \frac{\lambda - 1}{2}\rceil$ in $\Tilde{O}(D + \sqrt{n\lambda})$ rounds. Furthermore, they proved a lower bound of $\Tilde{\Omega}(D + \sqrt{\frac{n}{\lambda}})$ on the number of rounds required for such a decomposition.

    \paragraph*{Branching Random Walks in Networks.}
    The cover time of a random walk \cite{lawler2010random} on a graph is the time needed for a walk to visit each node at least once. Unfortunately, the expected value of this quantity is $\Omega(n\log n)$ even for a clique, making this primitive less useful in designing fast algorithms. Consequently, several attempts have been made to accelerate the cover time. Alon et al. \cite{alon2008many} proposed initiating multiple random walks from a single source. Subsequent work by Els\"asser and Sauerwald refined their bounds, demonstrating that $r$ random walks can yield a speed-up of $r$ times for many graph classes. Variations of multiple random walks have been applied in the CONGEST model to approximate the mixing time \cite{molla2017distributed}, perform leader election \cite{kutten2015sublinear, gilbert2018leader}, and evaluate network conductance \cite{fichtenberger2018two, batu2024all}.

    A branching random walk \cite{shi2015branching} (BRW) modifies the classical random walk by allowing nodes to emit multiple copies of a walk upon receipt, rather than simply relaying it. This branching behavior potentially leads to exponential growth in the number of walks traversing the graph, significantly reducing the cover time. 
    Gerraoui et al. \cite{guerraoui2023inherent} demonstrated that BRWs can enhance privacy by obscuring the source of gossip within a network. Recently, Aradhya et al. \cite{aradhya_et_al:LIPIcs.OPODIS.2024.36} employed BRWs to address permutation routing problems on subnetworks in the CONGEST model.

    Despite these applications, to the best of our knowledge, the branching random walk remains underexplored in distributed computing, and this work seeks to showcase its untapped potential.

    \section{Algorithm Overview}
    In this section, we provide a high-level overview of our algorithm, which consists of two major parts. First, we embed a virtual {\graph} atop the physical network $H$, and then we solve the problem on $G$. The reason we do this embedding is to transform an arbitrary expander into an almost-regular one. We explain the embedding procedure in Section \ref{sec:embedding random graph}, and from that point on, we focus solely on solving multi-message broadcast on an {\jgraph}. In Section \ref{sec:cobra intro}, we overview COBRA - the main building block of our algorithm for random graphs, and in Section \ref{sec:alg description} we describe the algorithm itself.

    \subsection{Embedding a Random Graph}
    \label{sec:embedding random graph}
    
   In this section, we describe how to embed a random graph atop a given expander. A related technique was used by Ghaffari et al.~\cite{ghaffari2017distributed}, where it suffices to embed a sparse {\jgraph} with expected degree $O(\log n)$. In our setting, however, it is crucial to preserve the \emph{minimum degree} of the host graph. To this end, we reuse Lemma~\ref{lem:ghaffari parallel random walks} from~\cite{ghaffari2017distributed}, but we also introduce additional ideas of node groups and rejection sampling to meet our stronger requirements.

    Embedding a graph \( G \) atop a host graph \( H \) involves creating virtual nodes $V(G)$ and establishing edges $E(G)$ between them so that
    \begin{itemize}
        \item Every virtual node $u \in V(G)$ is simulated by some physical node $\mathrm{host}(u) \in V(H)$.
        \item If a virtual node $u \in V(G)$ sends a message to $v \in V(G)$ along the edge $(u, v) \in E(G)$, this should be simulated by $\mathrm{host}(u)$ sending the same message to $\mathrm{host}(v)$ via some path in $H$. 
    \end{itemize}
    
    Our construction will guarantee that each round of communication in $G$ can be simulated in $O(\tmix(H) \cdot \log n)$ rounds in $H$ and that $\delta(G)$ will be close to $\delta(H)$. This gives us the following ``lifting'': if there is an algorithm that solves a problem in $f(k, \log (|V(G)|), \delta(G))$ rounds on $G$ for some function $f$, then there is an algorithm that solves a problem in essentially $f(k, \log n, \delta(H)) \cdot \tmix(H)\cdot \log n$ rounds on $H$. Formally, 

    \begin{theorem}
        \label{thm:embedding translation}
        Assume there exists an algorithm that solves multi-message broadcast on an {\jgraph} $G$ in $O(\frac{k}{\delta(G)}\cdot \log (|V(G)|) + \log^2(|V(G)|))$ rounds w.h.p. Then there exists an algorithm that solves multi-message broadcast on any graph $H$ on $n$ vertices in
        \begin{align*}
            O(\frac{k}{\delta(H)}\cdot \tmix(H) \cdot \log^2 n + \tmix(H)\cdot \log^3n)
        \end{align*} rounds w.h.p.
    \end{theorem}

    This round complexity consists of three terms: $O((\frac{k}{\delta(H)}\cdot \log n + \log^2n)\cdot \tmix \cdot \log n)$ for simulating an algorithm, $O(\tmix(H) \cdot \log^2 n)$ rounds for constructing an embedding of an {\jgraph} knowing $\tmix(H)$, and $O(\tmix(H) \cdot \log^2n)$ for estimating the $\tmix(H)$. Below, we describe the embedding and estimation parts.
    
    \paragraph*{Embedding Construction.} We now outline how to construct the embedding of an {\jgraph} atop the expander \( H\) preserving the minimal degree, assuming nodes have an estimate of \( \tmix \).
    
    We need to define a set of nodes of $G$. As an auxiliary concept, we start by defining \emph{sub-nodes}. First, rethink $H$ as being bidirectional, that is, replace every undirected edge $(u, v) \in E(H)$ with two directed edges $(u, v)$ and $(v, u)$. Now, we associate a \emph{sub-node} with each directed edge, and we say that a node $u \in V(H)$ simulates all the sub-nodes of its outgoing edges. Next, for each $u\in V(H)$, we arbitrarily group sub-nodes $u$ is simulating into \emph{groups} of size exactly $\delta(H)$. This might leave some residual sub-nodes that will not be assigned to any group; we call those \emph{inactive}, and others are called \emph{active}. Note that at most half of all sub-nodes can be inactive. Finally, we define nodes in $G$ to be the aforementioned groups, and we make a node $u \in V(H)$ simulate a node $v \in V(G)$ if it simulates the sub-nodes in the respective group.

    To establish edges in $G$, we launch lazy random walks of length $\tmix(H)$ from each active sub-node. The following lemma guarantees that those can be executed in parallel in only $O(\tmix(H)\cdot \log n)$ rounds w.h.p. 
    
    \begin{lemma}[Lemma 2.5~\cite{ghaffari2017distributed}]
    \label{lem:ghaffari parallel random walks}
    Let \( H = (V, E) \) be an \( n \)-node graph. Suppose we wish to perform \( T = n^{O(1)} \) steps of a collection of independent lazy random walks in parallel. If each node \( v \in V \) initiates at most \( d_H(v) \) walks, then w.h.p., the \( T \) steps of all walks can be performed in \( O(T \cdot\log n) \) rounds in a distributed setting.
    \end{lemma}

    If a random walk that started from a sub-node $a$, terminates at an active sub-node $b$, we call such a walk \emph{successful} and we propagate the ``success'' message back to $a$ by simply executing the walk in reverse. This establishes the edge between $\mathrm{group}(a) \in V(G)$ and $\mathrm{group}(b) \in V(G)$. If, on the other hand, the walk from $a$ terminates at an inactive node, we call such a walk \emph{failed} and we propagate the ``failure'' message back to $a$. Sub-nodes that received a ``failure'' message retry the same process again. The following lemma shows that only a few retries are needed.  

    \begin{lemma}
        After at most $O(\log n)$ retries, all sub-nodes will execute a successful random walk w.h.p. 
    \end{lemma}
    \begin{proof}
        By the definition of a mixing time, a random walk from a given sub-node is distributed (almost) uniformly among all the sub-nodes. Hence, given that at least half of all sub-nodes are active, a probability of success in one round is at least $1/2$ (minus a negligible term due to the ``almost'' uniformity). Therefore, the probability of not succeeding once after $100\log n$ trials is no more than $\frac{1}{n^{10}}$. Applying the union bound over all sub-nodes completes the proof.  
    \end{proof}

    Once paths are established, communication along \( (u, v) \in E(G) \) is simulated by routing a message along the corresponding walk path in \( H \), which, by Lemma~\ref{lem:ghaffari parallel random walks}, can be done in \( O(\tmix(H) \cdot \log n) \) rounds.
    
    We remark that the process described above is equivalent to sampling with replacement $\delta(H)$ neighbors for each node $v \in V(G)$ independently (almost) uniformly at random.  While this is not a canonical definition of an Erd\H{o}s--R\'enyi graph, the resulting distribution is equivalent, modulo a negligible inversely polynomial small probability.

    \paragraph*{On Distributed Estimation of Mixing Time.} Note that in order to implement this simulation, nodes do not need prior knowledge of $\tmix$. Instead, they can estimate the mixing time of the network using the decentralized algorithm of Kempe and McSherry~\cite{kempe2004decentralized}, which runs in \( O(\tmix \cdot \log^2n) \) rounds (without prior knowledge of \( \tmix \)). Their algorithm can be used to estimate the second principal eigenvalue \( \lambda \) of the transition matrix, which relates to the mixing time via the following inequality:

    \begin{theorem}[\cite{rivera2020lecture}]
    Given a graph on $n$ nodes with mixing time $\tmix$ and second principal eigenvalue $\lambda$, it holds that
    \[
    \left( \frac{1}{1 - \lambda} - 1 \right) \log n \leq \tmix \leq 2 \log n \cdot \frac{1}{1 - \lambda}.
    \]
    \end{theorem}

    Combining all together, we can prove Theorem \ref{thm:embedding translation}.

    \begin{proof}[Proof of Theorem \ref{thm:embedding translation}]
        First, estimate the mixing time $\tmix(H)$ of the graph $H$ using the approach from \cite{kempe2004decentralized}. This takes $O(\tmix(H)\cdot \log^2 n)$ rounds.

        Let the nodes of $H$ discover $\delta(H)$. This can be done through building a BFS tree in $H$ within $O(D(H)) = O(\tmix(H))$ rounds. 
        
        Knowing $\tmix(H)$ and $\delta(H)$, compute an embedding of an {\jgraph} $G(n', \frac{\delta(H)}{n'})$ as described above, where $n'$ is the number of sub-node groups in $H$. Make the source node for a multi-message broadcast in $G$ any node in $G$ that is simulated by the source of an original problem instance in $H$. Solve the problem in $G$ in $O(\frac{k}{\delta(G)}\cdot \log n + \log^2n) = O(\frac{k}{\delta(H)}\cdot \log n + \log^2n)$ rounds, simulating each round of $G$ in $O(\tmix(H)\cdot \log n)$ rounds in $H$. 
        
    \end{proof}

    \subsection{Coalescing-Branching Random Walk}
    \label{sec:cobra intro}
    We now overview the main primitive for our algorithm for an {\jgraph} -- the COalescing-BRAnching Random Walk (COBRA walk). COBRA walk was first introduced by Dutta et al. \cite{dutta2015coalescing} in their work ``Coalescing-Branching Random Walks on Graph'' \cite{dutta2015coalescing}, with subsequent refinements presented in \cite{cooper2016coalescing, mitzenmacher2018better, berenbrink2018tight}. The COBRA walk is a generalization of the classical random walk, defined as follows: At round $0$, a source node $s \in V$ possesses a token. At round $r$, each node possessing a token selects $\kappa$ of its neighbors uniformly at random, sends a token copy to each of them, and these neighbors are said to possess a token at round $r + 1$. Here, $\kappa$, referred to as the \emph{branching factor}, can be generalized to any positive real number \cite{cooper2016coalescing}. When $\kappa = 1$, the COBRA walk reduces to the classical random walk. From now on, we consider $\kappa$ to always be $2$. It is important to note that if a node receives multiple token copies in a round, it behaves as if it has received only one token; it will still choose $\kappa$ neighbors uniformly at random. This property, where received token copies coalesce at a node, gives the primitive its name.

    Cooper et al. \cite{cooper2016coalescing} studied the cover time of the COBRA walk on regular expanders and obtained the following theorem, which we use in our result

    \begin{theorem}[Cooper et al. \cite{cooper2016coalescing}]
        \label{thm:cooper cobra main}
        Let $G$ be a connected $n$-vertex regular graph. Let $\lambda_2$ be the second largest eigenvalue (in the absolute value) of the normalized adjacency matrix of $G$. Then after
        $
            O\left(\frac{\log n}{(1 - \lambda_2)^3}\right)
        $
        steps the COBRA walk covers $G$ with probability $1 - O(\frac{1}{n^2})$.
    \end{theorem}

    We point out that in \cite{cooper2016coalescing}, the bound on probability is $1 - O(\frac{1}{n})$, though the analysis, which is based on Chernoff bounds, can be adapted so that the probability is $1 - O(n^C)$ for any constant $C$ and the cover time is only multiplied by a constant.

    In the upcoming analysis of our result, we will make sure that $\lambda_2$ is no more than $\frac{13}{14}$ w.h.p. Let $C_T$ be a sufficiently large constant so that a COBRA walk covers a regular graph with $\lambda_2 \leq \frac{13}{14}$ with probability at least $1 - O(\frac{1}{n^2})$ in $C_T \cdot \log n$ rounds. From now on, we define $T$ to be $C_T\log n$.

    \subsection{Random Graph Algorithm Description}
    \label{sec:alg description}
    The algorithm to solve the multi-message broadcast on an {\graph} proceeds through the following steps:
    
    \begin{enumerate}
        \item \textbf{Building a BFS Tree and Gathering Information:}
        Nodes construct a BFS tree rooted at the source node $s$. Using the tree, every node learns the total number of nodes $|V|$, the minimum degree $\delta$, and the maximum degree $\Delta$. This step takes $O(D)$ rounds and does not require any prior knowledge of the graph topology.
    
        \item \textbf{Regularizing the Graph:}
        Each node $v$ adds $\Delta - \textit{deg}(v)$ self-loops to its adjacency list to make the graph regular. In our analysis, we will show that each node adds a relatively small number of self-loops. This operation is purely local and requires no communication between nodes.
    
        \item \textbf{Constructing Spanning Subgraphs through Multiple COBRA Walks:}  
        The source node $s$ initiates $\delta$ COBRA walks by creating $\delta$ tokens labeled $1$ through $\delta$. When a token from the $i$-th COBRA walk is sent along an edge $e:\ (u, v)$, nodes $u$ and $v$ mark $e$ as part of the $i$-th subgraph. Note that a single edge may belong to multiple subgraphs.
    
        When running multiple COBRA walks simultaneously, congestion can occur if a node attempts to send multiple tokens from different COBRA walks along the same edge in a single round. Since only one token can traverse an edge per round, this creates a bottleneck that needs to be managed. To address this, we organize the process into \emph{phases}, where each phase consists of $2$ rounds. In each phase, spanning rounds $\{2r, 2r+1\}$, every node $u$ distributes two tokens for each COBRA walk whose token(s) it received during the previous phase. Since there are $\delta$ COBRA walks, node $u$ could have received at most $p \leq \delta$ distinct tokens. Let these tokens be denoted by $t_1, \ldots, t_p$. Node $u$ then distributes each token twice, as follows: it enumerates its neighbors as $v_1, \ldots, v_\ell$, with $\ell \geq \delta$, generates two independent random permutations $\sigma_1, \sigma_2 \in S_\ell$, and sends token $t_i$ to neighbor $v_{\sigma_1(i)}$ in round $2r$ and to neighbor $v_{\sigma_2(i)}$ in round $2r+1$. Consequently, for any fixed COBRA walk $i$ and any node $u$, the token belonging to the $i$-th COBRA walk is sent to a neighbor of $u$ chosen uniformly at random, independently of the other tokens in that same COBRA walk. (Different cobra walks, however, need not be independent.) 
    
        This step completes in $T$ phases.

        \item[] \textbf{Parallel execution.} For the rest of the algorithm, we run protocols on all the constructed subgraphs in parallel. To achieve this despite potential congestion (recall that each edge may belong to multiple subgraphs), we again organize the execution into phases. Now, each phase spans $2T$ rounds, ensuring that messages sent along any shared edge are distributed across the protocols without conflict. Specifically, if a protocol would send a message along an edge in a particular round when executed independently, all such messages from different subgraphs are scheduled within the same phase. This phased execution allows all protocols to proceed in parallel while respecting the edge capacity.
    
        \item \textbf{Constructing Tree Packings:}  
        The source $s$ initiates a BFS on each subgraph to transform it into a tree. By the end of this step, the algorithm constructs a tree packing $\{T_i\}_{i \in [\delta]}$. This step takes the number of phases that is at most the maximal eccentricity of $s$ among all the subgraphs, that is at most $O(T)$ phases, and hence $2T \cdot O(T) = O(T^2)$ rounds.
    
        \item \textbf{Distributing Messages: }
        The source node $s$ evenly divides the set of messages $M$ across the $\delta$ trees, ensuring that each tree receives $\frac{k}{\delta}$ messages. These messages are then downcasted along the trees one by one. To broadcast $\frac{k}{\delta}$ messages in a single tree of diameter $O(T)$ one needs $O(\frac{k}{\delta} + T)$ rounds. Hence, doing it in parallel in all trees takes $O(T \cdot (\frac{k}{\delta} + T))$ rounds.
    \end{enumerate}

    \section{Proof}
    In this section, we formally state our results and their auxiliaries for the {\jgraph}. We start by providing a high-level overview.
    \subsection{Proof Outline}
    The proof proceeds in three main steps. First, we argue that making a random graph regular by adding self-loops does not significantly affect its expansion properties. To this end, we rely on a result of Hoffman et al.~\cite{hoffman2021spectral}, which shows that a random graph is a good expander with high probability. We also use standard Chernoff bound arguments to establish that random graphs are nearly regular, meaning the ratio between maximum and minimum degrees is close to one. Finally, we invoke a classical consequence of Weyl's inequality, which ensures that small perturbations to the diagonal of a matrix only slightly affect its eigenvalues. Together, these ingredients imply that regularizing the graph preserves its spectral expansion up to a small error. We discuss the details at the end of this section, namely in subsection \ref{sec:regularity + expansion}.
    
    Assuming the expansion properties remain, the second step is to prove that each individual COBRA walk successfully covers the entire network. The permutation trick allows us to claim that, though COBRA walks are not independent, their marginal distributions stay as if they were. Using this, Theorem \ref{thm:cooper cobra main} by Cooper et al., together with a union bound, guarantees that all COBRA walks cover the whole graph within $O(\log n)$ phases w.h.p.
    
    In the final step, we argue that the algorithm produces a tree packing with size $\delta(G)$, diameter $O(\log n)$, and weight $O(\log n)$. This tree packing allows for broadcasting all messages in $O(\log^2n + \log n \cdot \frac{k}{\delta(G)})$ rounds. This is optimal up to additive $\log^2n$ and multiplicative $\log n$, as $\frac{k}{\delta(G)}$ provides a natural lower bound for the optimal broadcast time: if there is a node with degree $\delta$, it needs at least $k/\delta$ rounds to receive $k$ messages.

    We now give the detailed proof starting from step 2, assuming step 1.

    \subsection{Success of Multiple COBRAs}
    \label{sec:cobra analysis}
    In this section, we demonstrate that multiple COBRA walks cover the graph fast, provided that it retains its expansion properties after adding self-loops. Formally, we assume the following lemma, which we prove in section \ref{sec:regularity + expansion}.

    \begin{restatable}{lemma}{lemGoodGraph}
        \label{lem:get good graph}
        With probability at least $1 - O(\frac{1}{n^2})$, for $p \geq \frac{C_p\log n}{n - 1}$, an {\graph} can be transformed into $G'$ by adding weighted self-loops to the nodes so that (1) $G'$ is regular, (2) $1 - \lambda_2(G') \geq \frac{1}{14}$.
    \end{restatable}
    
    Now, to see why multiple COBRAs do not congest, let us recall the permutation trick used in the algorithm. To be able to run multiple COBRA walks in parallel, we partition the execution into \emph{phases}, each lasting $2$ rounds. In each phase, spanning rounds ${2r, 2r+1}$, every node $u$ distributes two tokens for each COBRA walk whose token(s) it received during the previous phase. Since there are $\delta$ COBRA walks, node $u$ could have received at most $p \leq \delta$ distinct tokens. Let these tokens be denoted by $t_1, \ldots, t_p$. Node $u$ then distributes each token twice, as follows: it enumerates its neighbors as $v_1, \ldots, v_\ell$, with $\ell \geq \delta$, generates two independent random permutations $\sigma_1, \sigma_2 \in S_\ell$, and sends token $t_i$ to neighbor $v_{\sigma_1(i)}$ in round $2r$ and to neighbor $v_{\sigma_2(i)}$ in round $2r+1$.

    The claim below summarizes the properties resulting from the procedure described:
    \begin{claim}
        \label{clm:cobras don't interfere}
        For each COBRA walk, a phase corresponds exactly to a round in the execution where this COBRA walk runs in isolation. In particular, for a given COBRA walk, every token is sent independently of other tokens of this walk and to a uniformly chosen neighbor.
    \end{claim}

    To conclude the analysis of multi-COBRA's performance on a random graph, we combine Lemma \ref{lem:get good graph} and Claim \ref{clm:cobras don't interfere} and get the following lemma. 
    \begin{lemma}
        \label{lem:cobra compile}
        If the initial network graph is an {\graph} with $p \geq \frac{C_p\log n}{n}$, all COBRA walks cover the graph in $O(T)$ rounds with probability at least $1 - O(\frac{1}{n})$. 
    \end{lemma}
    \begin{proof}
        By Lemma \ref{lem:get good graph} we know that $G'$ - the graph we obtain from $G$ after adding self-loops - has $1 - \lambda_2(G') \geq \frac{1}{14}$ with probability at least $1 - O(\frac{1}{n^2})$.  
        Therefore, according to Theorem \ref{thm:cooper cobra main}, a COBRA walk succeeds to cover $G'$ in $O\left(\frac{\log n}{(1 - \lambda_2(G'))^3}\right) = O(T) = O(\log n)$ rounds with probability at least $1 - O(\frac{1}{n^2})$. And Claim \ref{clm:cobras don't interfere} tells us that while the COBRA walks are not independent, the marginal distribution of each individual process matches the distribution it would have under independent execution. Hence, we can apply a union bound and conclude that $\delta(G) \leq n$ COBRA walks cover $G'$ in $O(\log n)$ phases with probability at least $1 - O(\frac{1}{n})$.
    \end{proof}

    We remark that the analysis in \cite{cooper2016coalescing} is done for simple regular graphs (i.e., regular graphs not featuring self-loops). However, the arguments apply verbatim if self-loops are allowed, with symbols reinterpreted to mean the number of outgoing edges of a node instead of the number of neighbors.

    \subsection{Tree Packing and Broadcast}

    In this section, we show how to obtain a tree packing from multi-COBRA's edge assignment and describe how we use this tree packing to broadcast the messages. 

    The following two lemmas speak about the properties of the spanning graphs obtained via multi-COBRA.
    \begin{lemma}
        \label{lem:low edge weight}
        After multi-COBRA completes all $T$ phases, every edge of the graph belongs to at most $O(\log n)$ subgraphs.
    \end{lemma}
    \begin{proof}
        At each phase, each edge $(u, v)$ is assigned to at most $4$ subgraphs (two via tokens from $u$ to $v$, and two from $v$ to $u$), and there are $T = O(\log n)$ phases.
    \end{proof}

    \begin{lemma}
        \label{lem:low diameter}
        After multi-COBRA completes all $T$ phases, each subgraph has a diameter of $O(\log n)$.
    \end{lemma}
    \begin{proof}
        The multi-COBRA runs for $T = O(\log n)$ phases, and in each phase, we add to each subgraph only those nodes that are neighbors of the nodes already included. Consequently, the diameter of the subgraph increases by at most two per phase.
    \end{proof}

    In the rest of this section, we will analyze protocols that run on all the subgraphs in parallel. To achieve this parallelism despite potential congestion (recall that each edge may belong to multiple subgraphs), the execution is again organized into phases. Now, each phase spans $O(\log n)$ rounds, ensuring that messages sent along any shared edge are distributed across the protocols without conflict. Specifically, messages that would be sent on an edge in the same round by independently executed protocols are all scheduled within a single phase across subgraphs. This phased execution allows all protocols to proceed in parallel while respecting the edge capacity. As a result, the combined round complexity of the protocols increases by at most a factor of $O(\log n)$ compared to running an individual protocol.

    We are now ready to prove Theorem \ref{thm:building tree packing}.
    \thmTreePacking*
    \begin{proof}[Proof of Theorem \ref{thm:building tree packing}]
        The algorithm goes as follows. First, let nodes share the information of $n$, $\delta$ and $\Delta$. This can be done in $O(D)$ rounds by constructing a BFS tree. Next, every node $v$ adds $\Delta - \textit{deg}(v)$ self-loops. Then, parties run multi-COBRA for $T$ rounds that by Lemma \ref{lem:cobra compile} result with probability at least $1 - O(\frac{1}{n})$ in $\delta$ spanning subgraphs $\{S_i\}_{i \in [\delta]}$. By Lemma \ref{lem:low diameter}, those have diameter $O(\log n)$. Moreover, by Lemma \ref{lem:low edge weight}, every edge of the graph belongs to at most $O(\log n)$ subgraphs. 

        Now, we launch BFSs on all subgraphs in parallel to turn them into spanning trees. As discussed earlier in this section, this can be done in $O(\log n \cdot \max\limits_{i \in [\delta]}D(S_i)) = O(\log^2n)$ rounds. As a result of doing so, the weight of every edge could only have decreased, and the diameter of each subgraph at most doubled. 
    \end{proof}

    Having a tree packing with the properties described, we can prove Theorem \ref{thm:random main} by adding a final piece.

    \thmRandomMain*
    \begin{proof}[Proof of Theorem \ref{thm:random main}]
        First, build a tree packing from Theorem \ref{thm:building tree packing}. Then, $s$ evenly distributes messages among the obtained trees, so that each tree receives $\frac{k}{\delta}$ messages. After that, in each tree, nodes downcast corresponding messages. The algorithm for downcasting messages $\{m_1, \ldots, m_k\}$ from the root of a single tree works as follows. In the first round, the root sends the first message (\(m_1\)) to all its immediate children. In the second round, the children forward \(m_1\) to their respective children (the root's grandchildren), while the root simultaneously sends the second message (\(m_2\)) to its immediate children. This process continues iteratively: in each subsequent round, the root sends the next message (\(m_i\)) to its children, and all other nodes forward the message they received in the previous round to their respective children. This way, for $k'$ messages and a tree of diameter $H$ it takes $H + k' - 1$ rounds for every node to discover every message. 

        Multiplying by a congestion factor of $O(\log n)$, we get that the round complexity of broadcasting $k$ messages in $\delta(G)$ spanning trees of diameter $O(\log n)$ is 
        \begin{align*}
            O\left(\log^2n + \log n \cdot \frac{k}{\delta(G)}\right).
        \end{align*}
    \end{proof}

    \begin{remark}
        \label{rem:single source}
        In \cite{ghaffari2015distributed, chandra2024fast}, authors consider a problem where initially messages are spread over the network, that is every node possesses a subset of $M$. This seems like a more general version, however, when using a tree packing approach, these two problems are equivalent. The intuition is, nodes can first agree on the distribution of messages among trees, then upcast the messages to the root in their corresponding trees, and finally perform a downcast as described in our paper, all that in $\Tilde{O}(D(G) + \frac{k}{\delta(G)})$. For the full proof, a reader is invited to see the proof of Theorem 1 in \cite{chandra2024fast}.
    \end{remark}

    \subsection{Introducing Regularity while Maintaining Expansion}
    \label{sec:regularity + expansion}
    In this Section, we prove Lemma \ref{lem:get good graph} that states that graph's expansion properties are preserved after adding self-loops. We start by providing relevant concepts from spectral theory. 
    \begin{definition}
        Let $A$ be an $n \times n$ matrix with entries from $\mathbb{R}_{\geq 0}$ and let $D$ be a diagonal matrix such that $D_{ii} = \sum\limits_{j \in [n]}A_{ij}$. Assuming $D_{ii} > 0$ for all $i \in [n]$, let $\overline{A}$ denote a normalized version of $A$, i.e. $\overline{A} = D^{-1/2}AD^{-1/2}$.
    \end{definition}
    \begin{definition}
        Let $A$ be an $n\times n$ matrix. Define $\lambda_2(A)$ to be the second largest (in absolute value) eigenvalue of $A$. Let $G$ be an undirected multi-graph and $A$ be its weighted adjacency matrix. Define $\lambda_2(G)$ as $\lambda_2(\overline{A})$. 
    \end{definition}

    \begin{theorem}[Hoffman et al. \cite{hoffman2021spectral}]
        \label{thm:hoffman lambda2}
        For a positive constant $C$ and $p \geq \frac{C\log n}{n}$, consider an \graph. Then, with probability at least $1 - O(\frac{1}{n^{C - 1}})$, we have $\lambda_2(G) = O(\frac{1}{\sqrt{pn}})$.
    \end{theorem} 

    The following Lemma shows that with high probability, an {\graph} is almost regular. 

    \begin{lemma}
        \label{lem:deg dif}
        Let $p \geq \frac{C_p \log n}{n - 1}$. Then for an \graph, $\frac{\Delta(G)}{\delta(G)} \leq 1 + \frac{1}{7}$ with probability at least $1 - O(\frac{1}{n^2})$.
    \end{lemma}
    \begin{proof}
        Let us fix a vertex $v$ and consider the number of its incident edges. For each potential edge $e_i$, $i \in [n - 1]$ let us introduce an indicator variable $\chi_i$ which is equal to $1$ if the edge exists and to $0$ if it does not. The number of edges $v$ has is then $\sum\limits_{i \in [n - 1]}\chi_i$. The expectation of that is $p(n - 1)$, and applying Chernoff bounds, we get the following bounds on the degree of $v$
        \begin{align*}
            Pr\left[\textit{deg}(v) \geq (1 + \frac{1}{15})C_p\log n\right] \leq \exp\left(-\frac{C_p\log n}{675}\right) \leq \frac{1}{2n^3},
        \end{align*}
        and 
        \begin{align*}
            Pr\left[\textit{deg}(v) \leq (1 - \frac{1}{15})C_p\log n\right] \leq \exp\left(-\frac{C_p\log n}{675}\right) \leq \frac{1}{2n^3}.
        \end{align*}
        Now, taking union bound over all vertices, we conclude that for every vertex $v$ it holds that $(1 - \frac{1}{15})C_p\log n \leq \textit{deg}(v) \leq (1 + \frac{1}{15})C_p\log n$ with probability at least $1 - \frac{1}{n^2}$. Thus, with probability $1 - \frac{1}{n^2}$, we have that $\frac{\Delta(G)}{\delta(G)} \leq \frac{1 + \frac{1}{15}}{1 - \frac{1}{15}} = 1 + \frac{1}{7}$
    \end{proof}

    The next ingredient is to show that the slight perturbation of the diagonal elements of the matrix induces only a little change in its eigenvalues. 
    
    \begin{lemma}
        \label{lem:eigenvalue shift}
        Let $A$ be an $n \times n$ adjacency matrix of a connected graph and let $D$ be a diagonal matrix such that $D_{ii} = \sum\limits_{j \in [n]}A_{ij}$. Let $E$ be a $n \times n$ diagonal matrix such that $0 \leq E_{ii} \leq \varepsilon$ for some $0 < \varepsilon < 1$ and all $i \in [n]$.

        Then $\lambda_2(\overline{A + DE}) \leq \lambda_2(\overline{A}) + 6\varepsilon$.
    \end{lemma}
    \begin{proof}[Proof sketch.]
        The idea of the proof is to express $\overline{A + DE}$ as a sum of $\overline{A}$ and matrices with the small spectral norms, and then apply a corollary of Weyl's Theorem \cite{WeylsInequalityWiki}, that is, for $n\times n$ matrices $M_1$ and $M_2$, it holds that
        \begin{align*}
            |\lambda_2(M_1 + M_2) - \lambda_2(M_1)| \leq ||M_2||_2
        \end{align*}
        The full proof can be found in Appendix \ref{sec:appendix:proofs}.
    \end{proof}

    Finally, using Lemmas \ref{lem:deg dif} and \ref{lem:eigenvalue shift} alongside Theorem \ref{thm:hoffman lambda2}, we show that w.h.p., regularizing an {\graph} by adding self-loops for every node to reach $\Delta(G)$ does not ruin its expansion properties.
    
    \lemGoodGraph*
    \begin{proof}
        Let $A$ be the adjacency matrix of $G$, $D$ be the degree matrix of $G$ and $E$ be the $n \times n$ diagonal matrix with entries $E_{ii} = \frac{\Delta(G)}{\textit{deg}(v_i)} - 1$.  By the Lemma \ref{lem:deg dif}, with probability $1 - O(\frac{1}{n^2})$, $G$ has $\frac{\Delta}{\delta} \leq 1 + \frac{1}{7}$ and therefore, $E_{ii} \leq \frac{1}{7}$ for all $i \in [n]$ with probability $1 - O(\frac{1}{n^2})$.
    
        Now, to each vertex $v$ in $G$, add $\Delta - \textit{deg}(v)$ self-loops to obtain a graph $G'$. Clearly, $G'$ is $\Delta(G)$-regular. The adjacency matrix of $G'$ will then be $A + DE$ and hence, applying Lemma \ref{lem:eigenvalue shift} we deduce that $\lambda_2(G') \leq \lambda_2(G) + \frac{6}{7}$.
        
        By the Theorem \ref{thm:hoffman lambda2}, we know that with probability $1 - O(\frac{1}{n^2})$, $\lambda_2(G) \leq \frac{C}{\sqrt{p(n - 1)}}$ for some constant $C$, which for large enough $n$ is less than $\frac{1}{14}$, thus $\lambda_2(G') \leq \frac{13}{14}$.
    \end{proof}

\section{Sketching the terrain}
In this section, we share insights on the multi-message broadcast problem and why it is difficult to solve optimally. First, in Section \ref{sec:np-hardness}, we sketch the proof of its NP-hardness in the centralized setting. Then, in Section \ref{sec:bad lbs}, we argue that one can not hope to design an algorithm for a general graph that completes in $\tilde{O}(D(G) + \frac{k}{\delta(G)})$ rounds, which implies a need for some fundamentally new techniques.

\subsection{NP-Hardness}
\label{sec:np-hardness}
We prove that determining the optimal number of rounds for the multi-message broadcast problem in CONGEST is NP-hard in the centralized setting. To do that, we reduce the Set splitting problem to the multi-message broadcast.
\begin{definition}[Set splitting problem]
     Given a family $F$ of subsets of a finite set $S$, decide whether there exists a partition of $S$ into two subsets $S_1$, $S_2$ such that all elements of $F$ are split by this partition, i.e., none of the elements of $F$ is completely in $S_1$ or $S_2$.
\end{definition}

In our reduction, for simplicity of presentation, we allow edges to have arbitrary bandwidth instead of $\Tilde{O}(1)$, since, as we show, this can be simulated in CONGEST. In the reduction, the initial set $S$ corresponds to the set $M$ of messages, and $s$ has two dedicated children to which it can send $n_1$ and $n_2$ messages, respectively with $n_1 + n_2 = k$, simulating the splitting. Deciding how to split messages between these two children is the only ``smart'' choice an algorithm should make; all other nodes are just forwarding messages they receive. For the full version of the proof, please see Appendix \ref{sec:appendix:np-hardness}.

\subsection{Straightforward Lower Bounds are not Enough}
\label{sec:bad lbs}

\begin{figure}[ht]
\centering
\begin{minipage}[t]{0.68\linewidth}
  \vspace{0pt}
    It is tempting to argue for an approximation factor of an algorithm by comparing its round complexity to two straightforward lower bounds: $D(G)$ and $\frac{k}{\textit{minCut}(G)}$. 
    Unfortunately, those are not sufficient as there is an instance (see Figure \ref{fig:no simple lower bounds}) where $D(G) = O(1)$ and $\frac{k}{minCut(G)} = O(1)$, but the optimal answer is $\Omega(\sqrt{k})$. As for NP-hardness, we consider a more general model where edges might have arbitrary bandwidth, but we show that this can be simulated in CONGEST. For details, please see Appendix \ref{sec:appendix:Straightforward Lower Bounds are not Enough}.
\end{minipage}\hfill
\begin{minipage}[t]{0.28\linewidth}
  \vspace{0pt}
  \centering
  \includegraphics[width=\linewidth]{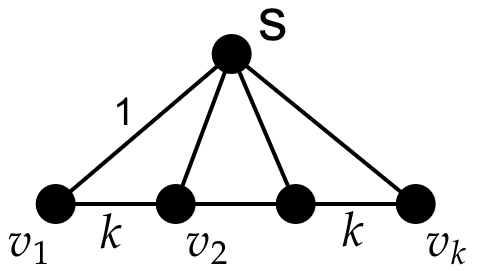}
  \caption{An example where diameter and minimum cut are not telling.
  Edge labels denote bandwidth.}
  \label{fig:no simple lower bounds}
\end{minipage}
\end{figure}

\section{Conclusion}
In this work, we explored the multi-message broadcast problem within the CONGEST model. We presented an algorithm that achieves universal optimality, up to polylogarithmic factors, for networks modeled as random graphs and extended our results through a lifting technique to general expander graphs, paying an additional factor of mixing time.

Our study introduces several promising avenues for future investigation. One intriguing direction involves developing distributed processes capable of rapidly covering graphs that also maintain composability, that is, efficiently supporting multiple simultaneous executions under congestion constraints. Another important direction is identifying graph properties beyond traditional metrics such as diameter and minimum cut, which would facilitate establishing tight lower bounds. Subsequently, algorithms matching these lower bounds could be devised, closing the open problem of multi-message broadcast.

\clearpage
\newpage
\bibliographystyle{plainurl}
\bibliography{refs}

\begin{thebibliography}{10}

\bibitem{ahlswede2000network}
Rudolf Ahlswede, Ning Cai, S-YR Li, and Raymond~W Yeung.
\newblock Network information flow.
\newblock {\em IEEE Transactions on information theory}, 46(4):1204--1216, 2000.

\bibitem{alon2008many}
Noga Alon, Chen Avin, Michal Koucky, Gady Kozma, Zvi Lotker, and Mark~R Tuttle.
\newblock Many random walks are faster than one.
\newblock In {\em Proceedings of the twentieth annual symposium on parallelism in algorithms and architectures}, pages 119--128, 2008.

\bibitem{aradhya_et_al:LIPIcs.OPODIS.2024.36}
Vijeth Aradhya, Seth Gilbert, and Thorsten G\"{o}tte.
\newblock {Distributed Branching Random Walks and Their Applications}.
\newblock In Silvia Bonomi, Letterio Galletta, Etienne Rivi\`{e}re, and Valerio Schiavoni, editors, {\em 28th International Conference on Principles of Distributed Systems (OPODIS 2024)}, volume 324 of {\em Leibniz International Proceedings in Informatics (LIPIcs)}, pages 36:1--36:20, Dagstuhl, Germany, 2025. Schloss Dagstuhl -- Leibniz-Zentrum f{\"u}r Informatik.
\newblock URL: \url{https://drops.dagstuhl.de/entities/document/10.4230/LIPIcs.OPODIS.2024.36}, \href {https://doi.org/10.4230/LIPIcs.OPODIS.2024.36} {\path{doi:10.4230/LIPIcs.OPODIS.2024.36}}.

\bibitem{batu2024all}
Tu{\u{g}}kan Batu, Amitabh Trehan, and Chhaya Trehan.
\newblock All you need are random walks: Fast and simple distributed conductance testing.
\newblock In {\em International Colloquium on Structural Information and Communication Complexity}, pages 64--82. Springer, 2024.

\bibitem{berenbrink2018tight}
Petra Berenbrink, George Giakkoupis, and Peter Kling.
\newblock Tight bounds for coalescing-branching random walks on regular graphs.
\newblock In {\em Proceedings of the Twenty-Ninth Annual ACM-SIAM Symposium on Discrete Algorithms}, pages 1715--1733. SIAM, 2018.

\bibitem{brandt2016lower}
Sebastian Brandt, Orr Fischer, Juho Hirvonen, Barbara Keller, Tuomo Lempi{\"a}inen, Joel Rybicki, Jukka Suomela, and Jara Uitto.
\newblock A lower bound for the distributed lov{\'a}sz local lemma.
\newblock In {\em Proceedings of the forty-eighth annual ACM symposium on Theory of Computing}, pages 479--488, 2016.

\bibitem{censor2014distributed}
Keren Censor-Hillel, Mohsen Ghaffari, and Fabian Kuhn.
\newblock Distributed connectivity decomposition.
\newblock In {\em Proceedings of the 2014 ACM symposium on Principles of distributed computing}, pages 156--165, 2014.

\bibitem{chandra2024fast}
Shashwat Chandra, Yi-Jun Chang, Michal Dory, Mohsen Ghaffari, and Dean Leitersdorf.
\newblock Fast broadcast in highly connected networks.
\newblock In {\em Proceedings of the 36th ACM Symposium on Parallelism in Algorithms and Architectures}, SPAA '24, page 331–343, New York, NY, USA, 2024. Association for Computing Machinery.
\newblock \href {https://doi.org/10.1145/3626183.3659959} {\path{doi:10.1145/3626183.3659959}}.

\bibitem{chang2020deterministic}
Yi-Jun Chang and Thatchaphol Saranurak.
\newblock Deterministic distributed expander decomposition and routing with applications in distributed derandomization.
\newblock In {\em 2020 IEEE 61st Annual Symposium on Foundations of Computer Science (FOCS)}, pages 377--388. IEEE, 2020.

\bibitem{chung2001diameter}
Fan Chung and Linyuan Lu.
\newblock The diameter of sparse random graphs.
\newblock {\em Advances in Applied Mathematics}, 26(4):257--279, 2001.

\bibitem{chuzhoy2020packing}
Julia Chuzhoy, Merav Parter, and Zihan Tan.
\newblock On packing low-diameter spanning trees.
\newblock {\em arXiv preprint arXiv:2006.07486}, 2020.

\bibitem{cooper2016coalescing}
Colin Cooper, Tomasz Radzik, and Nicolas Rivera.
\newblock The coalescing-branching random walk on expanders and the dual epidemic process.
\newblock In {\em Proceedings of the 2016 ACM Symposium on Principles of Distributed Computing}, pages 461--467, 2016.

\bibitem{dutta2015coalescing}
Chinmoy Dutta, Gopal Pandurangan, Rajmohan Rajaraman, and Scott Roche.
\newblock Coalescing-branching random walks on graphs.
\newblock {\em ACM Transactions on Parallel Computing (TOPC)}, 2(3):1--29, 2015.

\bibitem{edmonds1972}
J.~Edmonds.
\newblock Edge-disjoint branchings.
\newblock In R.~Rustin, editor, {\em Combinatorial Algorithms}, pages 91--96. Algorithmics Press, New York, New York, 1972.

\bibitem{elkin2017distributed}
Michael Elkin and Shay Solomon.
\newblock Distributed approximate maximum matching.
\newblock {\em ACM Transactions on Algorithms (TALG)}, 13(1):1--27, 2017.

\bibitem{erd6s1960evolution}
Paul Erd6s and Alfr{\'e}d R{\'e}nyi.
\newblock On the evolution of random graphs.
\newblock {\em Publ. Math. Inst. Hungar. Acad. Sci}, 5:17--61, 1960.

\bibitem{fichtenberger2018two}
Hendrik Fichtenberger and Yadu Vasudev.
\newblock A two-sided error distributed property tester for conductance.
\newblock In {\em 43rd International Symposium on Mathematical Foundations of Computer Science (MFCS 2018)}. Schloss-Dagstuhl-Leibniz Zentrum f{\"u}r Informatik, 2018.

\bibitem{fragouli2004decentralized}
Christina Fragouli and Emina Soljanin.
\newblock Decentralized network coding.
\newblock In {\em Information Theory Workshop}, pages 310--314. IEEE, 2004.

\bibitem{gao2014arboricity}
Pu~Gao, Xavier P{\'e}rez-Gim{\'e}nez, and Cristiane~M Sato.
\newblock Arboricity and spanning-tree packing in random graphs with an application to load balancing.
\newblock In {\em Proceedings of the twenty-fifth annual ACM-SIAM symposium on Discrete algorithms}, pages 317--326. SIAM, 2014.

\bibitem{garay1998sublinear}
Juan~A Garay, Shay Kutten, and David Peleg.
\newblock A sublinear time distributed algorithm for minimum-weight spanning trees.
\newblock {\em SIAM Journal on Computing}, 27(1):302--316, 1998.

\bibitem{ghaffari2015distributed}
Mohsen Ghaffari.
\newblock Distributed broadcast revisited: Towards universal optimality.
\newblock In {\em International Colloquium on Automata, Languages, and Programming}, pages 638--649. Springer, 2015.

\bibitem{ghaffari2016distributed}
Mohsen Ghaffari.
\newblock Distributed mis via all-to-all communication.
\newblock In {\em Proceedings of the 2016 ACM Symposium on Principles of Distributed Computing (PODC)}, pages 113--122, 2016.

\bibitem{ghaffari2016distributedPlanar}
Mohsen Ghaffari and Bernhard Haeupler.
\newblock Distributed algorithms for planar networks ii: low-congestion shortcuts, mst, and min-cut.
\newblock In {\em Proceedings of the twenty-seventh annual ACM-SIAM symposium on Discrete algorithms}, pages 202--219. SIAM, 2016.

\bibitem{ghaffari2013distributed}
Mohsen Ghaffari and Fabian Kuhn.
\newblock Distributed minimum cut approximation.
\newblock In {\em International Symposium on Distributed Computing}, pages 1--15. Springer, 2013.

\bibitem{ghaffari2017distributed}
Mohsen Ghaffari, Fabian Kuhn, and Hsin-Hao Su.
\newblock Distributed mst and routing in almost mixing time.
\newblock In {\em Proceedings of the ACM Symposium on Principles of Distributed Computing}, pages 131--140, 2017.

\bibitem{ghaffari2018new}
Mohsen Ghaffari and Jason Li.
\newblock New distributed algorithms in almost mixing time via transformations from parallel algorithms.
\newblock {\em arXiv preprint arXiv:1805.04764}, 2018.

\bibitem{gilbert2018leader}
Seth Gilbert, Peter Robinson, and Suman Sourav.
\newblock Leader election in well-connected graphs.
\newblock In {\em Proceedings of the 2018 ACM Symposium on Principles of Distributed Computing}, pages 227--236, 2018.

\bibitem{guerraoui2023inherent}
Rachid Guerraoui, Anne-Marie Kermarrec, Anastasiia Kucherenko, Rafael Pinot, and Sasha Voitovych.
\newblock On the inherent anonymity of gossiping.
\newblock {\em arXiv preprint arXiv:2308.02477}, 2023.

\bibitem{haeupler2022hop}
Bernhard Haeupler, Harald R{\"a}cke, and Mohsen Ghaffari.
\newblock Hop-constrained expander decompositions, oblivious routing, and distributed universal optimality.
\newblock In {\em Proceedings of the 54th Annual ACM SIGACT Symposium on Theory of Computing}, pages 1325--1338, 2022.

\bibitem{ho2011universal}
Tracey Ho, Sidharth Jaggi, Svitlana Vyetrenko, and Lingxiao Xia.
\newblock Universal and robust distributed network codes.
\newblock In {\em 2011 Proceedings IEEE INFOCOM}, pages 766--774. IEEE, 2011.

\bibitem{ho2003benefits}
Tracey Ho, Ralf Koetter, Muriel Medard, David~R Karger, and Michelle Effros.
\newblock The benefits of coding over routing in a randomized setting.
\newblock In {\em IEEE international symposium on information theory}, pages 442--442, 2003.

\bibitem{hoffman2021spectral}
Christopher Hoffman, Matthew Kahle, and Elliot Paquette.
\newblock Spectral gaps of random graphs and applications.
\newblock {\em International Mathematics Research Notices}, 2021(11):8353--8404, 2021.

\bibitem{jaggi2005polynomial}
Sidharth Jaggi, Peter Sanders, Philip~A Chou, Michelle Effros, Sebastian Egner, Kamal Jain, and Ludo~MGM Tolhuizen.
\newblock Polynomial time algorithms for multicast network code construction.
\newblock {\em IEEE Transactions on Information Theory}, 51(6):1973--1982, 2005.

\bibitem{kempe2004decentralized}
David Kempe and Frank McSherry.
\newblock A decentralized algorithm for spectral analysis.
\newblock In {\em Proceedings of the thirty-sixth annual ACM symposium on Theory of computing}, pages 561--568, 2004.

\bibitem{kutten2015sublinear}
Shay Kutten, Gopal Pandurangan, David Peleg, Peter Robinson, and Amitabh Trehan.
\newblock Sublinear bounds for randomized leader election.
\newblock {\em Theoretical Computer Science}, 561:134--143, 2015.

\bibitem{lawler2010random}
Gregory~F Lawler and Vlada Limic.
\newblock {\em Random walk: a modern introduction}, volume 123.
\newblock Cambridge University Press, 2010.

\bibitem{li2003linear}
S-YR Li, Raymond~W Yeung, and Ning Cai.
\newblock Linear network coding.
\newblock {\em IEEE transactions on information theory}, 49(2):371--381, 2003.

\bibitem{mitzenmacher2018better}
Michael Mitzenmacher, Rajmohan Rajaraman, and Scott Roche.
\newblock Better bounds for coalescing-branching random walks.
\newblock {\em ACM Transactions on Parallel Computing (TOPC)}, 5(1):1--23, 2018.

\bibitem{molla2017distributed}
Anisur~Rahaman Molla and Gopal Pandurangan.
\newblock Distributed computation of mixing time.
\newblock In {\em Proceedings of the 18th International Conference on Distributed Computing and Networking}, pages 1--4, 2017.

\bibitem{nash1961edge}
C~St~JA Nash-Williams.
\newblock Edge-disjoint spanning trees of finite graphs.
\newblock {\em Journal of the London Mathematical Society}, 1(1):445--450, 1961.

\bibitem{palmer2001spanning}
Edgar~M Palmer.
\newblock On the spanning tree packing number of a graph: a survey.
\newblock {\em Discrete Mathematics}, 230(1-3):13--21, 2001.

\bibitem{peleg2000distributed}
David Peleg.
\newblock {\em Distributed computing: a locality-sensitive approach}.
\newblock SIAM, 2000.

\bibitem{rivera2020lecture}
Nicolás Rivera, John Sylvester, Luca Zanetti, and Thomas Sauerwald.
\newblock Lecture 10: Random walks on graphs.
\newblock \url{https://www.cl.cam.ac.uk/teaching/1920/Probablty/materials/Lecture10.pdf}, 2020.
\newblock Lecture notes, University of Cambridge.

\bibitem{rozhovn2020polylogarithmic}
V{\'a}clav Rozho{\v{n}} and Mohsen Ghaffari.
\newblock Polylogarithmic-time deterministic network decomposition and distributed derandomization.
\newblock In {\em Proceedings of the 52nd Annual ACM SIGACT Symposium on Theory of Computing}, pages 350--363, 2020.

\bibitem{shi2015branching}
Zhan Shi et~al.
\newblock {\em Branching random walks}, volume 2151.
\newblock Springer, 2015.

\bibitem{swamy2013asymptotically}
Vasuki~Narasimha Swamy, Srikrishna Bhashyam, Rajesh Sundaresan, and Pramod Viswanath.
\newblock An asymptotically optimal push--pull method for multicasting over a random network.
\newblock {\em IEEE transactions on information theory}, 59(8):5075--5087, 2013.

\bibitem{tutte1961problem}
William~Thomas Tutte.
\newblock On the problem of decomposing a graph into n connected factors.
\newblock {\em Journal of the London Mathematical Society}, 1(1):221--230, 1961.

\bibitem{WeylsInequalityWiki}
{Wikipedia contributors}.
\newblock {Weyl's inequality}.
\newblock \url{https://en.wikipedia.org/wiki/Weyl%27s_inequality}, 2025.
\newblock [Accessed: 17-May-2025].

\bibitem{wu2004comparison}
Yunnan Wu, Philip~A Chou, and Kamal Jain.
\newblock A comparison of network coding and tree packing.
\newblock In {\em International Symposium onInformation Theory, 2004. ISIT 2004. Proceedings.}, page 143. IEEE, 2004.

\end{thebibliography}
\newpage

\appendix

\section{Technical definitions}
\label{sec:appendix: techincal defintions}
The mixing time of a Markov chain with transition matrix \( P \) and stationary distribution \( \pi \) can be defined as:
    
\begin{definition}[Mixing time]
\[
\tmix= \min \left\{ t : \max_x \left\| \frac{P^t(x, \cdot)}{\pi} - 1 \right\|_{2,\pi} \leq n^{-100} \right\},
\]
where
\[
\left\| \frac{P^t_x}{\pi} - 1 \right\|_{2,\pi} = \sqrt{ \sum_y \left( \frac{P^t(x,y)}{\pi(y)} - 1 \right)^2 \pi(y) }.
\]
\end{definition}

In the context of random walks on graphs, the state space corresponds to the edges of the graph. At each step, the walk moves — equiprobably — to one of the edges adjacent to the current edge. One can show that the resulting Markov chain admits a uniform stationary distribution.

\begin{remark}
To ensure that a Markov chain has a stationary distribution, it must be aperiodic. This is commonly enforced by making the walk \emph{lazy}, i.e., staying in the same state with probability \( 1/2 \).
\end{remark}

\begin{definition}[Conductance]
Let \(G = (V, E)\) be an undirected graph, and let \(\deg(v)\) denote the degree of vertex \(v\).  For any nonempty subset \(S \subset V\), define
\[
\partial S \;=\; \{\,\{u,v\}\in E : u\in S,\;v\notin S\},
\qquad
\mathrm{vol}(S) \;=\; \sum_{u\in S} \deg(u).
\]
The \emph{conductance} of the cut \((S, S^c)\) is
\[
\Phi(S) \;=\;
\frac{\lvert \partial S \rvert}{\min\bigl\{\mathrm{vol}(S),\,\mathrm{vol}(S^c)\bigr\}}.
\]
The \emph{conductance} of the graph \(G\) is then
\[
\Phi(G) \;=\;
\min_{\substack{S \subset V \\ 0 < \mathrm{vol}(S)\le \frac12\mathrm{vol}(V)}}
\Phi(S).
\]
\end{definition}

\begin{definition}[Chernoff bounds]
    Let \( X = \sum_{i=1}^n X_i \) be the sum of independent Bernoulli random variables with \( \mathbb{E}[X] = \mu \). Then, the following Chernoff bounds hold:

\begin{align*}
\Pr\left(X \leq (1 - \delta)\mu\right) &\leq e^{-\delta^2 \mu / 2},\ 0 \leq \delta\\
\Pr\left(X \geq (1 + \delta)\mu\right) &\leq e^{-\delta^2 \mu / (2 + \delta)},\ 0 \leq \delta\\
\Pr\left(|X - \mu| \geq \delta \mu\right) &\leq 2e^{-\delta^2 \mu / 3},\  0 \leq \delta \leq 1.
\end{align*}
\end{definition}

\section{Straightforward Lower Bounds are not Enough}
\label{sec:appendix:Straightforward Lower Bounds are not Enough}

\begin{figure}
    \centering
    \includegraphics[width=0.48\textwidth]{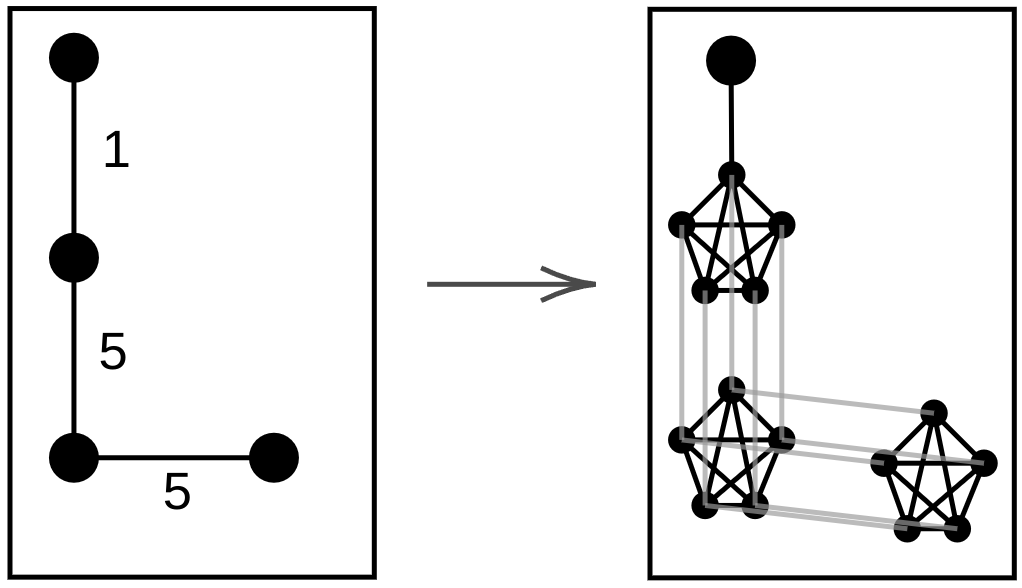}
    \caption{An example of mapping a graph with arbitrary bandwidths to a graph suitable for CONGEST.}
    \label{fig:simulating big bandwidth}
\end{figure}
In this section, it will be more comfortable for us to consider a more general model than CONGEST, namely the model where edges have arbitrary bandwidth. To transform a graph with arbitrary bandwidths to a graph with all bandwidths equal to $1$, we do the following. The source $s$ corresponds to a single node in the new graph. For a node $v \neq s$ in the original graph, let $B$ denote the maximal bandwidth of its adjacent edges. In the new graph, node $v$ then corresponds to a clique of $B$ nodes. We call this clique a $v$-clique. If in the original graph nodes $v \neq s$ and $u \neq s$ were connected by an edge of bandwidth $b$, we pick (arbitrary) $b$ nodes in $v$-clique, $b$ nodes in $u$-clique, and draw $b$ edges between picked nodes to establish a perfect matching. For every edge $(s, u)$ of bandwidth $b$, we connect the new source with $b$ arbitrary nodes of the $u$-clique. We call the resulting graph \emph{the corresponding CONGEST graph}.

\begin{claim}
    \label{clm:small diameter change}
    Consider the original graph $G$ and its corresponding CONGEST graph $G'$. Then\\
    $D(G) \leq D(G') \leq 2D(G) + 1$.
\end{claim}
\begin{proof}[Proof idea]
    The first inequality is straightforward. We prove the second inequality by induction on the length of the path, that is if there is a path in $G$ from $u$ to $v$ of length $l$, then for any nodes $u'$ and $v'$ in $u$-clique and $v$-clique respectively, there is a path between $u'$ and $v'$ in $G'$ of length $2l + 1$.
\end{proof}

\begin{claim}
    \label{clm:small minuct change}
    Consider the original graph $G$ and its corresponding CONGEST graph $G'$. Then\\
    $\min\{minCut(G), \min\limits_{v \in V(G) \setminus \{s\}}  \text{size of the $v$-clique} - 1 \}\leq minCut(G') \leq minCut(G)$.
\end{claim}
\begin{proof}
    The second inequality is straightforward. For the first inequality, note that each cut of $G'$ either cuts some clique or does not. In case it does not, it corresponds to a cut in $G$ and has the same size. In case it does, it is at least the size of the induced cut for that clique, which is at least $\min\limits_{v \in V(G) \setminus \{s\}}  \text{size of the $v$-clique} - 1$.
\end{proof}

\begin{claim}
    \label{clm:small round complexity change}
    Consider the original graph $G$ and its corresponding CONGEST graph $G'$. Together with a set $M$ of messages, they define a multi-message broadcast problem in generalized CONGEST and CONGEST, respectively. Let $\mathrm{OPT}(G)$ and $\mathrm{OPT}(G')$ denote the optimal round complexities for $G$ and $G'$ respectively. Then $\mathrm{OPT}(G) \leq \mathrm{OPT}(G')$.
\end{claim}
\begin{proof}
    Consider an execution $E'$ for $G'$ which achieves $\mathrm{OPT}$. We claim that we can build an execution $E$ for $G$, such that for every round $r$ of $E'$ and for every $v \in V(G)$, after round $r$ in $E$, $v$ knows all the messages that the nodes of $v$-clique know after round $r$ in $E'$. To do so, consider a round $r$ and some $v \in V(G)$. Let us say that nodes in $v$-clique in $E'$ in round $r$ receive messages $M_1$ from the $u_1$-clique, messages $M_2$ from the $u_2$-clique, and so forth for all neighboring cliques. Then, in $E$ $u_i$ sends $M_i$ to $v$ satisfying the invariant. Note that $u_i$ can do this in terms of the bandwidth by construction of the corresponding CONGEST graph, and in terms of knowing $M_i$ by the invariant.  
\end{proof}

\begin{figure}[ht]
\centering
\begin{minipage}[t]{0.68\linewidth}
  \vspace{0pt}
  We now give an example of an instance of a problem where $D(G) = O(1)$ as well as $\frac{k}{minCut(G)} = O(1)$, but the optimal round complexity is $\Omega({\sqrt{k}})$. The graph we consider is the corresponding CONGEST graph $G'$ to the graph $G$ depicted in Figure \ref{fig:no simple lower bounds appendix}. 

  First, note that $D(G) = 2$, hence by Claim \ref{clm:small diameter change}, $D(G') = O(1)$. Second, notice that $minCut(G) = k$ and the minimal maximal bandwidth of an edge adjacent to some node in $V(G) \setminus \{s\}$ is equal to $k$, therefore, by Claim \ref{clm:small minuct change}, $k - 1 \leq minCut(G') \leq k$. Finally, by Claim \ref{clm:small round complexity change}, $\mathrm{OPT}(G') \geq \mathrm{OPT}(G)$, where $\mathrm{OPT}$ is the optimal round complexity. Therefore, it is sufficient to show that $\mathrm{OPT}(G) = \Omega(\sqrt{k})$. 
\end{minipage}\hfill
\begin{minipage}[t]{0.28\linewidth}
  \vspace{0pt}
  \centering
  \includegraphics[width=\linewidth]{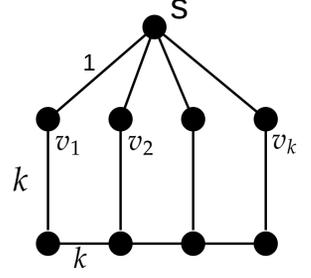}
  \caption{An example graph $G$ where diameter and minimum cut are not telling. Here, edge labels denote bandwidth.}
  \label{fig:no simple lower bounds appendix}
\end{minipage}
\end{figure}

We claim that $\Omega(\sqrt{k})$ rounds are needed for $v_1$ only to get to know $M$ (become saturated). For the sake of contradiction, assume that we can saturate $v_1$ in $\leq \sqrt{k} - 1$ rounds. That means that it can be saturated without using the edges $(s, v_{\sqrt{k} + 1})$ and $(v_{\sqrt{k} + 1}, v_{\sqrt{k} + 2})$. But if we remove those edges, $minCut(s, v_1) \leq \sqrt{k}$, implying that the number of rounds needed is at least $\frac{k}{\sqrt{k}} = \sqrt{k}$, a contradiction.

\section{NP-hardness}
\label{sec:appendix:np-hardness}
To show the NP-hardness of a multi-message broadcast problem, we will also use a generalization of CONGEST that allows for arbitrary edge bandwidth, though this time the construction is different. In this section, we will consider a specific layered graph with layers induced by the distance from $s$. In that graph, all edges connect nodes of consecutive layers. This graph has arbitrarily large bandwidths assigned to its edges, so we transform it into a graph with unit bandwidths by doing the following. For each node $v$ on layer $0 < l < \text{max layer}$, we create a group of $n$ nodes called $v_{out}$, where $n$ denotes the number of messages (we change the notation due to reduction). Then, for every edge $(u, v)$ of the original graph, where $u$ belongs to the previous layer ($l - 1$), if that edge has bandwidth $b \leq n$, we create a group of $b$ nodes called $v_{u-in}$ and we connect arbitrary $b$ nodes of $u_{out}$ 1 to 1 to node of $v_{u-in}$. For every $u$, we connect every node of $v_{u-in}$ to every node of $v_{out}$. For node $s$, we replace it with a new sink $s'$ and create a $K_{n, n}$ with its first half called $s_{in}$ and its second half called $s_{out}$. We then connect $s'$ to all the nodes in $s_{in}$, and we connect all the nodes of $s_{out}$ to the $in$-s of the nodes $s$ is connected to in the original graph in a way described above. For all the nodes of the last layer, we keep them a single node and draw all the incoming edges to this node. We call the resulting graph of this transformation the \emph{transformed} graph. Please see Figure \ref{fig:transform for np} for an example. 
\begin{figure}[H]
    \centering
    \includegraphics[width=0.5\linewidth]{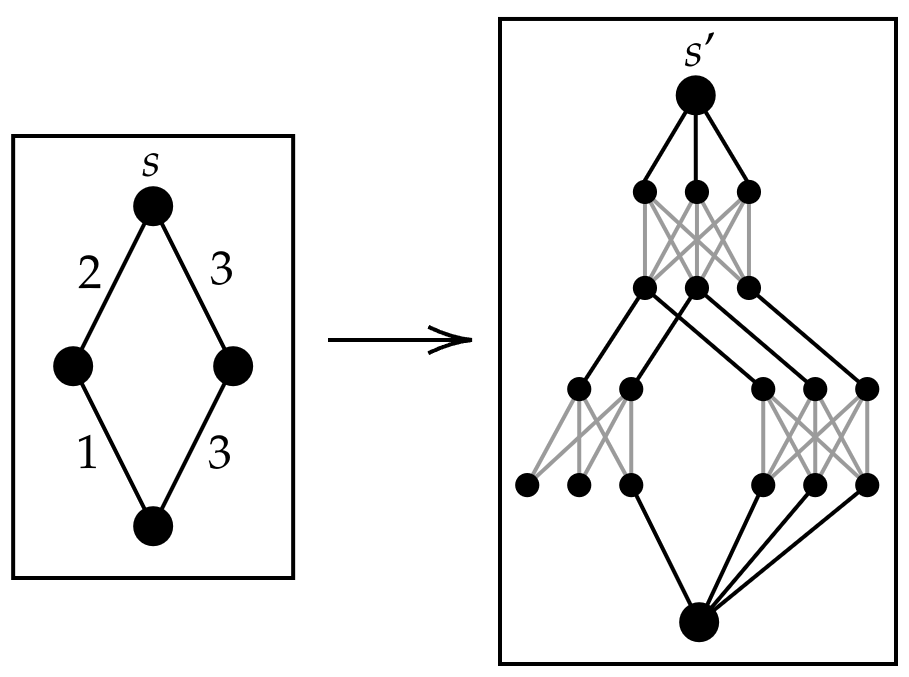}
    \caption{Example of transforming a layered graph with arbitrary bandwidths into a graph suitable for CONGEST. Here the number of messages $n$ is $3$.}
    \label{fig:transform for np}
\end{figure}

\begin{claim}
    \label{clm:np equivalence}
    Consider a layered graph $G$ with arbitrary bandwidths and $G'$ being its transformed version. Denote $l$ to be the depth of $G$ and let $T$ be the set of nodes in $G$ in layer $l$. Then it is possible to saturate all nodes in $T$ in $l$ rounds in $G$ if and only if it is possible to saturate all nodes in $T$ in $2l + 1$ rounds in $G'$.
\end{claim}
\begin{proof}
    Consider an execution $E$ for $G$ in which all nodes in $T$ are saturated in $l$ rounds. We build an execution $E'$ for $G'$ that satisfies the following invariant: for $0 \leq r < l$, after $2r + 2$ rounds of $E'$, for every node $v \in V(G)$ such that $v$ is in layer $t \leq r$, nodes in $v_{out}$ know the same set of messages in $E'$ as $v$ knows in $E$ after round $r$. For $r = 0$, we make $s'$ send all messages to $s_{in}$ (a distinct message to each node) and $s_{in}$ to relay those messages to $s_{out}$. Then, if in round $r > 0$ in $E$ $u$ sends $v$ $b$ messages, $u_{out}$ send $v_{u-in}$ those $b$ messages and then $v_{u-in}$ relay those to $v_{out}$. In the final $l$-th round of $E$, nodes of $G$ send messages to $t_i \in T$. This can be simulated in $E'$ within one round, making it $2(l - 1) + 2 = 2l$ rounds to reach $v_{out}$ for all $v$-s in layer $l - 1$ and $1$ more round to saturate $T$.

    The proof of the other direction proceeds analogously, maintaining the invariant that every node $v \in V(G)$ in $E$ after $r$ rounds knows all the messages that $v_{out}$ knows in $E'$ after $2r + 2$ rounds. 
\end{proof}

\begin{theorem}
    The multi-message broadcast problem is NP-hard in a centralized setting.
\end{theorem}
\begin{proof}
    We reduce the Set splitting problem: given a family $F$ of subsets of a finite set $S$, decide whether there exists a partition of $S$ into two subsets $S_1$, $S_2$ such that all elements of $F$ are split by this partition, i.e., none of the elements of $F$ is completely in $S_1$ or $S_2$.

    Denote $n := |S|$, $m := |F|$. We start creating a reduction graph by creating a source node $s$ and assigning it a set of messages corresponding to elements in $S$: $\{m_1, \ldots, m_n\}$. We also create $n$ nodes $v_1, \ldots, v_n$ with edges $(s, v_i)$ of bandwidth $1$. Intuitively, we want every $v_i$ to hold $m_i$ after the first round. 

    We create nodes that correspond to the elements of $F$: $F_1, \ldots F_m$ and we draw an edge $(v_i, F_j)$ of bandwidth $1$ iff $S[i] \in F[j]$. This way, after round two, $F_i$ will possess messages that correspond to the elements of $F[i]$. 

    With a slight abuse of notation, we introduce two other nodes, namely $S_1$ and $S_2$, that intuitively correspond to a partition of $S$. We focus on solving the set partition problem for the given size of the parts, i.e., $|S_1| = n_1$ and $|S_2| = n_2$ with $n_1 + n_2 = n$. Obviously, this version is also NP-complete. We draw an edge $(s, S_1)$ of bandwidth $n_1$ and $(s, S_2)$ of bandwidth $n_2$. We want $S_1$ and $S_2$ to be in layer $2$, so we introduce intermediate nodes on those edges whose role will simply be to relay messages.  

    Now, we introduce nodes $u_{i, 1}, u_{i, 2}$ for $i \in [m]$. We draw following edges: $(F_i, u_{i,1})$ with bandwidth $|F_i|$, $(S_1, u_{i, 1})$ with bandwidth $n_1$. Similarly, for $u_{i, 2}$ we draw $(F_i, u_{i, 2})$ with bandwidth $|F_i|$ and $(S_2, u_{i, 2})$ with bandwidth $n_2$. Intuitively, $u_{i, 1}$ serves the meaning of the union of $F[i]$ and $S_1$. 

    We introduce nodes $t_{i, 1}$ and $t_{i, 2}$ for $i \in [m]$. For $i \in [m]$ we draw an edge $(u_{i, 1}, t_{i,1})$  of bandwidth $n$ and, and this is the crux of the reduction, an edge $(s, t_{i, 1})$ of bandwidth $n - n_1 \mathbf{- 1}$ and of length $4$ (with $3$ intermediate nodes). The idea here is that $t_{i, 1}$ can be saturated after round $4$ if and only if it receives more than $n_1$ messages from $u_{i, 1},$ implying $F[i] \not\subset S_1$. Similarly, we do for $S_2$. See the resulting construction in Figure \ref{fig:np construction}.

    \begin{figure}
        \centering
        \includegraphics[width=0.9\linewidth]{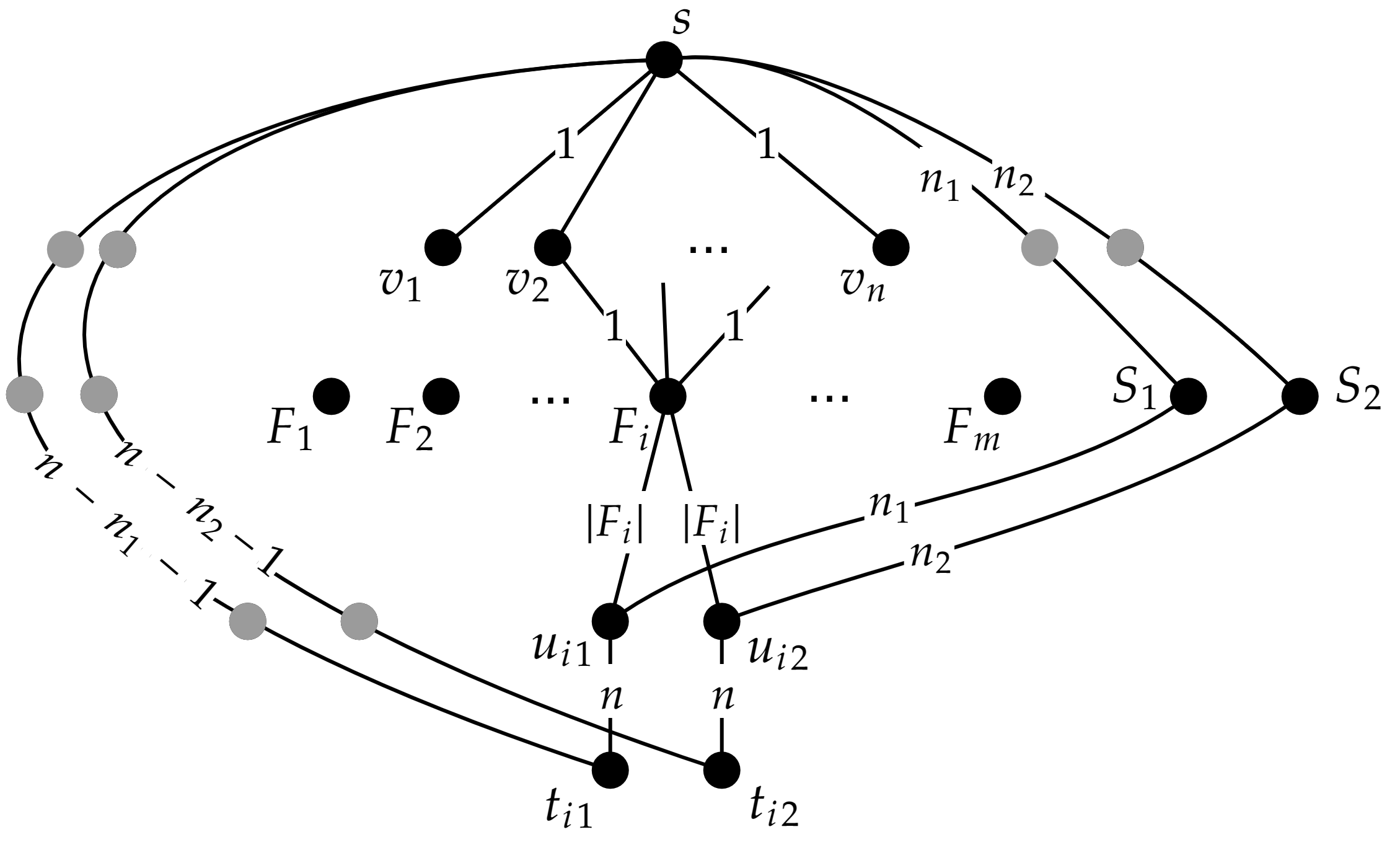}
        \caption{Graph $G$ with arbitrary bandwidths for which it is NP-hard to optimally solve multi-message broadcast. Some nodes are depicted in gray since they serve no other purpose but to layer the graph, and in reasonable executions, should only relay messages.}
        \label{fig:np construction}
    \end{figure}

    If we now consider a transformed graph $G'$, we want to focus on saturating nodes in $T = \{t_{11}, t_{12}, \ldots, t_{m1}, t_{m2}\}$, though in multi-message broadcast problem the goal is to saturate \emph{all} nodes. To account for that, for each node $v \in V(G') \setminus T$, we will make sure that it can be saturated in $9$ rounds. We do that by introducing a path of length $9$ and bandwidth $n$ from $s$ to $v$. In particular, each such path has $6$ intermediate layers of $n$ nodes each. Each node in the first layer is connected to each node in $s_{out}$. Each node in layer $1 < l \leq 6$ is connected to each node in layer $l - 1$, and $v$ is connected to each node in layer $6$. This way, we obtain the graph $G''$. Note that introducing these additional paths does not help saturate $T$ in fewer than $9$ rounds, that is, $T$ can be saturated in $9$ rounds in $G''$ iff it can be saturated in $9$ rounds in $G'$.
    
    These observations, combined with Claim \ref{clm:np equivalence} allow us to establish the following sequence of equivalent statements ($\Leftrightarrow$ denotes equivalence):\\
    (I) The set splitting for $S$ and $F_1, \ldots, F_m$ is possible $\Leftrightarrow$\\
    (II) Saturating $T$ in $G$ in $4$ rounds is possible $\Leftrightarrow$\\
    (III) Saturating $T$ in $G'$ in $9$ rounds is possible $\Leftrightarrow$\\
    (IV) Solving the multi-message broadcast in $9$ rounds in $G''$ is possible.

    The equivalence of (II) and (III) is Claim \ref{clm:np equivalence}. The equivalence of (III) and (IV) is discussed above. Hence, leaving the details of those unspecified, we focus on the informative part - the equivalence of (I) and (II). 

    First, assume it is possible to split $S$, and this splitting is $S_1$ and $S_2$. Then we claim it is possible to saturate $T$ in $G$ in $4$ rounds. To do so, let $s$ send $\{m_i \mid i \in S_1\}$ to $S_1$ and $\{m_i \mid i \in S_2\}$ to $S_2$. Also, let it send $m_i$ to $v_i$ and to $t_{ij}$, $i \in [m]$, $j \in \{1, 2\}$, $s$ sends $S \setminus (F_i \cup S_j)$. After that, nodes only relay the messages they have to further layers. Now we claim that after round $4$, all nodes in $T$ are saturated. Indeed, for instance, $t_{i1}$ will receive $F_i \cup S_1 \cup (S \setminus (F_i \cup S_1)) = S$, the main point being that since $F_i \subsetneq S_1$, $|F_i \cup S_1| > |S_1| = n_1$, therefore $|S \setminus (F_i \cup S_1)| \leq n -  n_1 - 1$ and $s$ can send it whole. 

    Now, assume we can saturate $T$ in $G$ in $4$ rounds. This implies that every $u_{i1}$ in round $3$ holds more than $n_1$ messages, implying that messages held by $F_i$ are not a strict subset of messages held by $S_1$ in round $2$. Analogously, it holds for $F_i$ and $S_2$. This means, there is (possibly non-injective) mapping $\phi$ of $\{v_1, \ldots, v_n\}$ into $S$ so that $\forall i \in [m], j \in \{1, 2\}$ it holds that $(\ast)$ $\bigcup\limits_{l \in F[i]}\phi(v_l) \subsetneq S_j$. Note that by making $\phi$ injective (and thus bijective) by iteratively taking a colliding pair $x, y$ ($\phi(x) = \phi(y)$) and assigning $y$ to the so far uncovered element, we cannot break $\ast$. Therefore, we can assume that $\phi$ (i.e., distribution of messages across $v_i$) is bijective, which gives a solution to the splitting problem up to permuting the elements. 

\end{proof}

\section{Technical proofs}
\label{sec:appendix:proofs}
\begin{proof}[Proof of Lemma \ref{lem:eigenvalue shift}]
    Let $A' = A + DE$ and let $D'$ be a diagonal matrix such that $D'_{ii} = \sum\limits_{j \in [n]}A'_{ij}$. Note that $D' = D + DE$ and hence $(D')^{-1/2} = (D)^{-1/2}(I + E)^{-1/2}$. Entries of the $(I + E)^{-1/2}$ are of the form $\frac{1}{\sqrt{1 + E_{ii}}} \geq \frac{1}{\sqrt{1 + \varepsilon}} \geq \sqrt{1 - \varepsilon} \geq 1 - \varepsilon$. Therefore, we can denote $(I + E)^{-1/2}$ with $I - E'$ where $E'$ is a diagonal matrix with entries $0 \leq E'_{ii} \leq \varepsilon$. Now
        \begin{align*}
            \overline{A'} =& (D')^{-1/2}A'(D')^{-1/2}\\
                          =& (D^{-1/2} - D^{-1/2}E')(A + DE)(D^{-1/2} - D^{-1/2}E')\\
                          \overset{\ast}{=}& D^{-1/2}AD^{-1/2} - D^{-1/2}AD^{-1/2}E' + E - EE' - E'D^{-1/2}AD^{-1/2} + \\& E'D^{-1/2}AD^{-1/2}E' - E'E + E'EE'\\
                          =& \overline{A} - \overline{A}E' - E'\overline{A} + E'\overline{A}E' + E - 2EE' + E'EE'
        \end{align*}
        where to obtain $\ast$ we used the fact that diagonal matrices commute. 

        From Weyl's theorem, we conclude that 
        \begin{align*}
            \lambda_2(\overline{A'}) - \lambda_2(\overline{A}) &\leq ||- \overline{A}E' - E'\overline{A} + E'\overline{A}E' + E - 2EE' + E'EE'||_2\\
            \leq& ||\overline{A}E'||_2 + || E'\overline{A}||_2 + ||E'\overline{A}E' ||_2 + 2||EE'||_2 + ||E'EE' ||_2 \\
            \leq& ||\overline{A}||_2||E'||_2 + || E'||_2||\overline{A}||_2 + ||E'||_2||\overline{A}||_2||E' ||_2 +\\
            &2||E||_2||E'||_2 + ||E'||_2||E||_2||E' ||_2 \\
        \end{align*}

        Now recall that the spectral norm for a real-valued symmetric matrix is the biggest absolute value of its eigenvalues, hence $||\overline{A}||_2 = 1$ and $||E||_2 \leq \varepsilon$, $||E'||_2 \leq \varepsilon$. Thus
        \begin{align*}
            \lambda_2(\overline{A'}) - \lambda_2(\overline{A}) \leq \varepsilon + \varepsilon + \varepsilon^2 + 2 \varepsilon^2 + \varepsilon^3 \leq 6\varepsilon
        \end{align*}
\end{proof}

\end{document}